\documentclass[11pt]{article}
\usepackage{amsfonts}
\usepackage{amssymb}
\usepackage{amsmath}
\usepackage{graphicx}
\usepackage{enumerate}
\usepackage{hhline}
\usepackage[ruled,vlined,linesnumbered,commentsnumbered]{algorithm2e}
\usepackage{fullpage}
\usepackage{amsthm}
\usepackage{xcolor}
\usepackage{tikz}
\usetikzlibrary{arrows}

\usepackage{hyperref}

\def\eps{\varepsilon}
\def\prob{\mathrm{Prob}}
\newcommand{\expec}{\mathbb{E}}

\long\def\CProof #1\CQED{}

\def\blackslug{\hbox{\hskip 1pt \vrule width 4pt height 8pt
    depth 1.5pt \hskip 1pt}}
\def\QED{\quad\blackslug\lower 8.5pt\null\par}

\long\def\PPP#1{\noindent{\bf Proof:}{ #1}{\quad\blackslug\lower 8.5pt\null}}

\long\def\denspar #1\densend
{#1}

\newcommand{\tO}{\tilde{O}}
\newcommand{\DkS}{Densest $k$-Subgraph\xspace}

\newcounter{note}[section]

\newtheorem{lemma}{Lemma}[section]
\newtheorem{theorem}[lemma]{Theorem}
\newtheorem{observation}[lemma]{Observation}

\newtheorem{definition}[lemma]{Definition}

\newtheorem{remark}[lemma]{Remark}

\begin{document}

\title{The Densest $k$-Subhypergraph Problem 
}

\author{
Eden Chlamt\'a\v{c}
\thanks{Department of Computer Science, Ben Gurion University. Partially supported by ISF grant 1002/14.
Email: {\tt chlamtac@cs.bgu.ac.il.
}}
\and
Michael Dinitz
\thanks{
Department of Computer Science,
Johns Hopkins University.  Partially supported by NSF grants 1464239 and 1535887.
Email: {\tt mdinitz@cs.jhu.edu
}}
\and 
Christian Konrad
\thanks{
ICE-TCS, School of Computer Science, Reykjavik University.
Supported by Icelandic Research Fund grants 120032011 and 152679-051.
Email: {\tt christiank@ru.is}.
}
\and
Guy Kortsarz\thanks{
Computer Science Department, 
Rutgers University, Camden, NY, USA. 
Partially supported by NSF grants  1218620 and  1540547. 
Email: {\tt guyk@crab.rutgers.edu}}
\and
George Rabanca \thanks{Department of Computer Science,
The Graduate Center, CUNY, USA.  
Email: {\tt grabanca@gmail.com}}
}
\date{\empty}

\begin{titlepage}
\def\thepage{}
\maketitle
\begin{abstract}
The Densest $k$-Subgraph (D$k$S) problem, and its corresponding minimization problem Smallest $p$-Edge Subgraph (S$p$ES), have come to play a central role in approximation algorithms.  This is due both to their practical importance, and their usefulness as a tool for solving and establishing approximation bounds for other problems.  These two problems are not well understood, and it is widely believed that they do not an admit a subpolynomial approximation ratio (although the best known hardness results do not rule this out).

In this paper we generalize both D$k$S  and S$p$ES from graphs to hypergraphs.  We consider the Densest $k$-Subhypergraph problem (given a hypergraph $(V, E)$, find a subset $W\subseteq V$ of $k$ vertices so as to maximize the number of hyperedges contained in $W$) and define the Minimum $p$-Union problem (given a hypergraph, choose $p$ of the hyperedges so as to minimize the number of vertices in their union).  We focus in particular on the case where all hyperedges have size~$3$, as this is the simplest non-graph setting.  For this case we provide an $O(n^{4(4-\sqrt{3})/13 + \epsilon}) \leq O(n^{0.697831+\epsilon})$-approximation (for arbitrary constant $\epsilon > 0$) for Densest $k$-Subhypergraph and an $\tO(n^{2/5})$-approximation for Minimum $p$-Union.  We also give an $O(\sqrt{m})$-approximation for Minimum $p$-Union in general hypergraphs.  Finally, we examine the interesting special case of interval hypergraphs (instances where the vertices are a subset of the natural numbers and the hyperedges are intervals of the line) and prove that both problems admit an exact polynomial time solution on these instances. 
\end{abstract}

\end{titlepage}\pagenumbering {arabic}

\section{Introduction}

Two of the most important outstanding problems in approximation 
algorithms are the approximability of the \emph{Densest $k$-Subgraph} problem 
(D$k$S) and its minimization version, the \emph{Smallest $p$-Edge Subgraph} problem
(S$p$ES or min-D$k$S).  In D$k$S we are given as input a graph 
$G = (V, E)$ and an integer $k$, and the goal is to find a subset 
$V' \subseteq V$ with $|V'| = k$ which maximizes the number of edges in 
the subgraph of $G$ induced by $V'$.  In the minimization version, 
S$p$ES, we are given a lower bound $p$ 
on the number of required edges and
the goal is to find a set $V' \subseteq V$ of minimum size
so that the subgraph induced by $V'$ has at least $p$ edges.  These problems have proved to be extremely useful: for example, a variant of D$k$S was recently 
used to get a new cryptographic system
\cite{crypt}. The same variant of the D$k$S problem 
was shown to be central in understanding 
financial derivatives~\cite{der}.
The best-known algorithms for many other problems involve using an algorithm for \DkS or S$p$ES as a black box (e.g.~\cite{Nut10,GHNR10,DKN14}).

Despite decades of work, very little is actually known about these problems.
The first approximation ratio for D$k$S was 
$O(n^{2/5})$ \cite{kp93} and was devised in 1993. These days,
23 years later, the best known ratio for the 
\DkS is $O(n^{1/4 + \epsilon})$ for arbitrarily small constant $\epsilon > 0$~\cite{BCCFV10}, and the best known approximation for S$p$ES is $O(n^{3-2\sqrt{2} + \epsilon})$ for arbitrarily small constant $\epsilon > 0$~\cite{CDK12}.  Given the slow improvement over 23 years, it is widely believed 
that D$k$S and S$p$ES do not admit better than 
a polynomial approximation ratio. Furthermore, the existing approximation guarantees are tight assuming the recently conjectured hardness of finding a planted dense subgraph in a random graph (for certain parameters). However, there has been very little progress towards an actual proof of hardness of approximation.  It is clear that they are both NP-hard, but that is all that is known under the assumption that $P \neq NP$.  Under much stronger complexity assumptions it is known that they cannot be approximated better than some constant~\cite{Khot06,Fei02} or any constant~\cite{AAMMW11}, but this is still a long way from the conjectured polynomial hardness. 

Based on the believed hardness of D$k$S and S$p$ES, they have been used many times 
to give evidence for hardness of approximation.
For example, consider the {\em Steiner $k$-forest} problem
in which the input is an edge weighted 
graph, a collection of $q$ pairs $\{s_i,t_i\}_{i=1}^q$,
and a number $k<q$. The goal is to find a minimum cost subgraph
that connects at least $k$ of the pairs. 
It is immediate to see that S$p$ES is a special case of 
the Steiner $k$-forest problem\footnote{Given an instance $(G= (V, E), p)$ of S$p$ES, create an instance of Steiner $k$-Forest on a star with $V$ as the leaves, uniform weights, a demand pair for each edge in $E$, and $k=p$.}, and hence it seems highly unlikely
that the Steiner $k$-Forest problem admits a better than 
polynomial approximation ratios.

Given the interest in and importance of D$k$S and S$p$ES, it is somewhat surprising that there has been very little exploration of the equivalent problems in \emph{hypergraphs}. 
A hypergraph is most simply understood as a collection $E$ of subsets over a 
universe $V$ of vertices, where each $e \in E$ is called a \emph{hyperedge} (so graphs are the special case when each $e \in E$ has cardinality $2$).
 In general hypergraphs, the obvious extensions of D$k$S and S$p$ES are quite intuitive. In the \emph{Densest $k$-Subhypergraph} (D$k$SH) problem we are 
given a hypergraph $(V, E)$
and a value $k$, and the goal is to find a set $W\subseteq V$ of size $k$ that 
contains the largest number of hyperedges from $E$.  In the \emph{Minimum $p$-Union} (M$p$U) problem we are given a hypergraph 
and a number $p$, and the goal is to choose $p$ of the hyperedges to minimize the size of their union.

Clearly these problems are at least as hard as the associated problems in graphs, but how much harder are they?  Can we design nontrivial approximation algorithms?  Can we extend the known algorithms for graphs to the hypergraph setting?  Currently, essentially only lower bounds are known: Applebaum~\cite{Applebaum13} showed that they are both hard to approximate to within $n^{\epsilon}$ for some fixed $\epsilon > 0$, assuming that a certain class of one-way functions exist.  But it was left as an open problem to design any nontrivial upper bound (see footnote $5$ of~\cite{Applebaum13}).  

\subsection{Our Results}

In this paper we provide the first nontrivial upper bounds for these problems.  Let $n$ denote the number of vertices and $m$ denote the number of hyperedges in the input hypergraph.  Our first result is an approximation for Minimum $p$-Union in general hypergraphs:

\begin{theorem} \label{thm:MkU-main}
There is an $O(\sqrt{m})$-approximation for the Minimum $p$-Union problem.
\end{theorem}

We then switch our attention to the low rank case, since this is the setting closest to graphs.  In particular, we focus on the \emph{$3$-uniform} case, where all hyperedges have size at most $3$.  In this setting it is relatively straightforward to design an $O(n)$-approximation for Densest $k$-Subhypergraph, although even this is not entirely trivial (the optimal solution could have size up to $k^3$ rather than $k^2$ as in graphs, which would make the trivial algorithm of choosing $k/3$ hyperedges only an $O(n^2)$-approximation rather than an $O(n)$-approximation as in graphs).  We show that by very carefully combining a set of algorithms and considering the cases where they are all jointly tight we can significantly improve this approximation, obtaining the following theorem:

\begin{theorem} \label{thm:3DkH-main}
For every constant $\epsilon > 0$, there is an $O(n^{4(4-\sqrt{3})/13 + \epsilon}) \leq O(n^{0.697831+\epsilon})$-approximation for the Densest $k$-Subhypergraph problem on $3$-uniform hypergraphs.
\end{theorem}

Adapting these ideas to the minimization setting gives an improved bound for Minimum $p$-Union as well.

\begin{theorem} \label{thm:MkU-uniform}
There is an $\tO(n^{2/5})$-approximation for the Minimum $p$-Union problem on $3$-uniform hypergraphs.
\end{theorem}

It is worth noting that any $f$-approximation for D$k$SH can be used to give an $\tO(f)$-approximation for M$p$U (see Theorem~\ref{thm:DkHMkU}), so Theorem~\ref{thm:MkU-uniform} gives a significant improvement over this blackbox reduction from Theorem~\ref{thm:3DkH-main}.  

Finally, we define an interesting 
special case of Densest $k$-Subhypergraph and Minimum $p$-Union that can be solved exactly in polynomial time.  Suppose we have an \emph{interval hypergraph}: a hypergraph in which the vertices are a finite subset of $\mathbb{N}$ and each hyperedge is an interval of the real line (restricted to the vertices).  Then we show that a dynamic programming algorithm can be used to actually solve our problems.

\begin{theorem} \label{thm:interval-main}
Densest $k$-Subhypergraph and Minimum $p$-Union can be solved in polynomial time on interval hypergraphs.
\end{theorem}

\subsection{Related Work}
As discussed, the motivation for these problems mostly comes from the associated graph problems, which have been extensively studied and yet are still poorly understood.  The Densest $k$-Subgraph problem was introduced by Kortsarz and Peleg~\cite{kp93}, who gave an $O(n^{2/5})$ ratio
for the problem. Feige, Kortsarz and Peleg~\cite{fkp} improved the 
ratio to $O(n^{1/3-\epsilon})$ 
for $\epsilon$ that is roughly $1/60$. 
 The current best-known approximation for D$k$S is $O(n^{1/4 + \epsilon})$ for arbitrarily small constant $\epsilon > 0$, due to Bhaskara et al.~\cite{BCCFV10}.  For many years the minimization version, S$p$ES, was not considered separately, and it was only relatively recently that the first separation was developed: building on the techniques of~\cite{BCCFV10} but optimizing them for the minimization version, Chlamt\'a\v{c}, Dinitz, and Krauthgamer~\cite{CDK12} gave an $O(n^{3-2\sqrt{2} + \epsilon})$-approximation for S$p$ES for arbitrarily small constant $\epsilon > 0$.  
 
 While defined slightly differently, D$k$SH and M$p$U were introduced earlier by Applebaum~\cite{Applebaum13} in the context of cryptography: he showed that if certain one way functions exist (or that certain pseudorandom generators exist) then D$k$SH is hard to approximate within $n^{\epsilon}$ for some constant $\epsilon > 0$.  Based on this result, D$k$SH and M$p$U were used to prove hardness for other problems, such as the $k$-route cut problem~\cite{CMVZ16}.  To the best of our knowledge, though, there has been no previous work on algorithms for these problems.

\subsection{Organization}

We begin in Section~\ref{sec:prelims} with some preliminaries, showing the basic relationships between the problems.  In Section~\ref{sec:MkU} we give our $O(\sqrt{m})$-approximation for M$p$U in general hypergraphs.  We then focus on small-rank hypergraphs, giving an $O(n^{4/5})$-approximation for D$k$SH on $3$-uniform hypergraphs in Section~\ref{sec:DkSH3}, which we then improve to roughly $O(n^{0.698})$ in Section~\ref{sec:3DkSH-improved}. 
  We follow this in Section~\ref{sec:MpU3} with our improved bound for M$p$U on $3$-uniform hypergraphs.  Finally in Section~\ref{sec:interval} we show how to solve both problems exactly in polynomial time on interval hypergraphs.  We conclude in Section~\ref{sec:open} with some open questions for future work.  

\section{Preliminaries and Notation} \label{sec:prelims}

A \emph{hypergraph} $H = (V, E)$ consists of a set $V$ (the vertices) together with a collection $E \subseteq 2^{V}$ (the hyperedges), where each hyperedge is a subset of $V$.  We will typically use $n = |V|$ and $m=|E|$ to denote the number of vertices and hyperedges respectively.  The \emph{degree} of a vertex in a hypergraph is the number of hyperedges which contain it.  Given a subset $V' \subseteq V$, the subhypergraph of $H$ \emph{induced} by $V'$ is $H[V'] = (V', E_H)$ where $E_H = \{e \in E : e \subseteq V'\}$.  We say that $H$ is \emph{$\alpha$-uniform} if $|e| = \alpha$ for all $e \in E$, and that the \emph{rank} of $H$ is $\max_{e \in E} |e|$ (i.e.~the smallest $\alpha$ such that all edges have cardinality at most $\alpha$).  A hyperedge $e$ is \emph{covered} by a set of vertices $V'$ if $e \subseteq V'$.  

Given a graph $G = (V, E)$ and a vertex $v \in V$, we use $\Gamma_G(v)$ to denote the set of nodes adjacent to $v$, and for a subset $V' \subseteq V$ we let $\Gamma_G(V') = \cup_{v \in V'} \Gamma(v)$.  If $G$ is clear from context, we will sometimes drop the subscript.

The main problems that we will consider are the following.

\begin{definition}
Given a hypergraph $H = (V, E)$ and an integer $k$, the  \emph{Densest $k$-Subhypergraph} problem (D$k$SH) is to find a set $V' \subseteq V$, with $|V'| = k$, such that the number of edges in $H[V']$ is maximized.
\end{definition}

\begin{definition}
Given a hypergraph $H = (V, E)$ and an integer $p$, the \emph{Minimum $p$-Union} problem (M$p$U) is to find a set $E' \subseteq E$, with $|E'| = p$, such that $|\cup_{e \in E'} e|$ is minimized.  
\end{definition}

Note that on $2$-uniform hypergraphs, these two problems are the classic graph problems D$k$S and S$p$ES respectively.  

A special class of hypergraphs that we will consider are \emph{interval hypergraphs}, defined as follows.

\begin{definition}
$H = (V, E)$ is an \emph{interval hypergraph} if $V$ is a finite subset of $\mathbb{N}$ and for each $e \in E$ there are values $a_e, b_e \in \mathbb{N}$ such that $e = \{i \in V : a_e \leq i \leq b_e\}$.
\end{definition}

\subsection{Relationship Between Problems} \label{sec:relationship}

We begin by proving some relatively straightforward relationships between the two problems.  We first make the obvious observation that a solution for one problem implies a solution for the other.  

\begin{observation} \label{obs:poly-time-reduction}
If there exists a polynomial time algorithm that solves the Densest $k$-Subhypergraph problem for any $k$ on a hypergraph $H$, then there exists a polynomial time algorithm that solves the Minimum $p$-Union problem on the hypergraph $H$.  Similarly, if there is an algorithm that solves M$p$U on $H$, then there is an algorithm that solves D$k$SH on $H$.
\end{observation}

The relationship is not quite so simple when we are reduced to approximating the problems, but it is 
 relatively straightforward to show that a relationship still exists.  This is given by the following lemma, which will also prove to be useful later.  

\begin{lemma}\label{lem:MkU-component} If there exists an algorithm which in a hypergraph $H$ containing a subhypergraph with $k$ vertices and $p$ hyperedges finds a subhypergraph $(V', E')$ with $|V'| \leq fk$ and $|E'| \geq |V'| p/(kf)$, we can get an $O(f\log p)$-approximation for Min $p$-Union.
\end{lemma}

Since any $f$-approximation algorithm for Densest $k$-Subhypergraph satisfies the conditions of the lemma, as an immediate corollary we get the following:

\begin{theorem} \label{thm:DkHMkU}
If there is an $f$-approximation for Densest $k$-Subhypergraph, then there is an $O(f \log p)$-approximation for Minimum $p$-Union.
\end{theorem}

\begin{proof}[Proof of Lemma~\ref{lem:MkU-component}]
Let $(H = (V, E), p)$ be an instance of Minimum $p$-Union, and let $\mathcal A$ be an algorithm as described in the lemma. 
We assume without loss of generality that we know the number of nodes $k$ in the optimal solution (since we can just try all possibilities for $k$), and hence that there exists a set $V^* \subseteq V$ with $|V^*| = k$ such that $V^*$ covers at least $p$ hyperedges.  Initialize $E' = \emptyset$, and consider the following algorithm for Minimum $p$-Union that repeats the following until $|E'| \geq p$.
\begin{enumerate}
\item Let $V' = \mathcal A(H, k)$, and let $E''$ be the hyperedges of $H$ covered by $V'$.
\item Let $E' \leftarrow E' \cup E''$.
\item Remove $E''$ from $H$ (remove only the edges, not the corresponding vertices).
\end{enumerate}

We claim that this is an $\tO(f)$-approximation for Minimum $p$-Union.  
Indeed, suppose at iteration~$i$ we added $x_i$ vertices, and that at the beginning of the iteration, we had already added $p-p_i$ edges to the solution. In particular, that means that at least $p_i$ of the original hyperedges contained in $V^*$ were not yet removed. This then implies that the number of edges added in iteration $i$ was at least $x_i\cdot p_i/(kf)$. Thus the number of edges we still need to add after iteration $i$ is $p_{i+1}\leq p_i-x_i\cdot p_i/(kf)=p_i(1-x_i/(kf))$.
Thus by induction, after $t$ iterations, the number of hyperedges we need to add is bounded by $$p_{t+1}\leq p\prod_{i=1}^t(1-x_i/(kf))\leq p\exp\left(-\sum_{i=1}^tx_i/(kf)\right).$$ Thus, as soon as the total number of vertices added exceeds $kf\ln p$ for the first time, the number of edges will exceed $p$. Since the last iteration adds at most $kf$ vertices, we are done.
\end{proof}

A standard argument also shows a (more lossy) reduction in the other direction.

\begin{theorem} \label{thm:MkUDkH}
If there is an $f$-approximation for Minimum $p$-Union on $\alpha$-uniform hypergraphs, then there is an $O(f^{\alpha})$-approximation for Densest $k$-Subhypergraph on $\alpha$-uniform hypergraphs (when $\alpha = O(1)$).
\end{theorem}

\section{Minimum $p$-Union in General Hypergraphs} \label{sec:MkU}

Given a hypergraph $H = (V, E)$, in this section we work with the bipartite {\em incidence graph} 
$G = (E, V, F)$ of $H$, where $F = \{ (e, v) \in E \times V \, : \, v \in e \}$.
Solving M$p$U on $H$ corresponds to finding a subset $E' \subseteq E$ 
of $p$ vertices in $G$ of minimum vertex expansion, i.e., 
$E'$ such that $|\Gamma_G(E')|$ is minimized.

Our algorithm requires a subroutine that returns a subset of vertices of minimum expansion (without the cardinality bound on the set).  In other words, we need a polynomial-time algorithm {\sc Min-Exp(G)} which returns a subset of $E$ so that 
\begin{equation*}
\frac{|\textsc{Min-Exp}(G)|}{|\Gamma_G(\textsc{Min-Exp}(G))|} \ge \frac{|E'|}{|\Gamma_G(E')|}, 
\end{equation*}
for every subset $E' \subseteq E$.  

Minimally expanding subsets of this kind have previously been used (e.g. in \cite{kr13,gkk12}) in communication
settings where computation time is disregarded, but in our context we need a polynomial-time algorithm. In Appendices~\ref{app:flow} and~\ref{sec:charikar} we give two different algorithms for doing this.  The first, in Appendix~\ref{app:flow}, uses a reduction to network flows.  The second, in Appendix~\ref{sec:charikar}, is based on a straightforward adaptation of a linear programming approach for the graph case due to Charikar~\cite{Charikar}.  In order to simplify the presentation, we will for the rest of the section assume that we have such an algorithm and will defer them to the appendices.

In the following, for subsets $E' \subseteq E$ and $V' \subseteq V$, we denote the induced subgraph of 
$G$ by vertex set $E' \cup V'$ by $G[E', V']$.

In the first phase, our algorithm (Algorithm~\ref{alg:min-union}) iteratively adds vertices $E''$
to an initially empty set $E'$ until $E'$ exceeds the size $p - \sqrt{m}$. The set $E''$ is a
minimally expanding subset in the induced subgraph $G[E \setminus E', V]$. If $E''$ is large so that
$|E' \cup E''| > p$, then an arbitrary subset of $E''$ is added to $E'$ so that $E'$ has the 
desired size $p$. Then, in the second phase, we add the $k - |E'|$ vertices of $E \setminus E'$ of smallest degree to
$E'$ (ties broken arbitrarily), and the algorithm returns set $E'$.

\begin{algorithm}
 \KwData{Bipartite input graph $G = (E, V, F)$ with $m=|E|$, $n=|V|$, parameter $p$}
 $E' \gets \{ \}$\;
 \Repeat{$|E'| \ge p - \sqrt{m}$}{
  $E'' \gets \textsc{Min-Exp}( G[E \setminus E', V])$\;
  \eIf{$|E'| + |E''| \le p$}{
   $E' \gets E' \cup E''$\;  
  }
  {
   Add arbitrary $p - |E'|$ nodes from $E''$ to $E'$\; \label{line:058}
  }
 }
 $E'' \gets $ subset of $p - |E'|$ nodes of $E \setminus E'$ of smallest degree\; \label{line:201}
 $E' \gets E' \cup E''$\;
 \KwRet{$E'$}\;
 \caption{$2\sqrt{m}$-approximation algorithm for the Minimum $p$-Union problem \label{alg:min-union}}
\end{algorithm}

\begin{theorem}
 Algorithm~\ref{alg:min-union} is a $(2 \sqrt{m})$-approximation algorithm for M$p$U.
\end{theorem}

\begin{proof}
Let $OPT \subseteq E$ be an optimal solution and let $r = |\Gamma_G(OPT)|$. 
Let $E'_i$ denote the set $E'$ in the beginning of the $i$th iteration of the repeat loop. Suppose 
that the algorithm runs in $l$ rounds. Then, $E'_{l+1}$ is the set $E'$ after the last iteration of the loop, 
but before the nodes selected in Line~\ref{line:201} are added. 

Consider an arbitrary iteration $i \le l$ and let $E'' \gets \textsc{Min-Exp}( G[E \setminus E'_i, V] )$ 
as in the algorithm. Note that by the condition of the loop, we have $|E_i'| \le p - \sqrt{m}$. 
Furthermore, we have
$$\frac{|E''|}{|\Gamma_G(E'')|} \ge \frac{|OPT \setminus E'_i|}{|\Gamma_G(|OPT \setminus E'_i|)|} \ge \frac{k - |E'_i|}{r},$$
since $E''$ is a set of minimum expansion. Then, 
$$|\Gamma_G(E'')| \le \frac{|E''| r}{p - |E_i'|} \le \frac{|E''| r}{p - p + \sqrt{m}} = \frac{|E''| r}{\sqrt{m}}.$$

Thus, we have $|\Gamma_G(E'_{i+1})| \le |\Gamma_G(E'_{i})| + \frac{|E''| r}{\sqrt{m}}$ (note that this 
inequality also captures the case when only a subset of $E''$ is added to $E'$ in Line~\ref{line:058}). 
Now, note that the sets $E''$ of
any two different iterations are disjoint and thus the sizes of the sets $E''$ of the different iterations 
sum up to at most $m$. We thus obtain the bound:
$$|\Gamma_G(E'_{l+1})| \le \frac{m r}{\sqrt{m}} = \sqrt{m} r.$$

In phase two, we select at most $\sqrt{m}$ vertices $E''$ of minimum degree in $G[E \setminus E', V]$.
Clearly, the maximum degree of these vertices is at most $r$ (if it was larger, then $|\Gamma_G(OPT)|$ 
would be larger as well) and thus $|\Gamma_G(E'')| \le \sqrt{m}r$. The neighborhood of the returned set of 
our algorithm is hence at most $2 \sqrt{m} r$ which gives an approximation factor of $2\sqrt{m}$.
\end{proof}

\section{Densest $k$-Subhypergraph in $3$-uniform hypergraphs}\label{sec:DkSH3}

In this section, we consider the Densest $k$-Subhypergraph problem in $3$-uniform hypergraphs. 
We develop an $O(n^{4/5})$-approximation algorithm here, and show in Section~\ref{sec:3DkSH-improved}
 how to improve
the approximation factor to $O(n^{0.697831+\epsilon})$, for any $\epsilon > 0$, by replacing one of 
our subroutines with an algorithm of Bhaskara et al. \cite{BCCFV10}.

Throughout this section, let $H=(V,E)$ be the input $3$-uniform hypergraph. Let $K \subseteq V$ denote an optimal
solution, i.e., a subset of vertices such that $H[K]$ is a densest $k$-subhypergraph. The average degree of 
$H[K]$ is denoted by $d = 3|E(H[K])|/k$. We say that a hyperedge is optimal if it is contained in $H[K]$.

\subsection{Overview of our Algorithm}
Let $K_1 \subseteq V$ be a set of $k/3$ vertices of largest degree (ties broken arbitrarily), 
$\Delta$ the minimum degree of a node in $K_1$, and $H' = H[V \setminus K_1]$. Note that the 
maximum degree in $H'$ is $\Delta$.

Suppose first that 
at least half of the optimal hyperedges contain at least one vertex of $K_1$. Then the following lemma shows that
we can easily achieve a much better approximation than we are aiming for:
\begin{lemma}\label{lem:greedy2}
Suppose that at least half of the optimal hyperedges contain a vertex of $K_1$. Then we can achieve an $O(n^{1/4+\eps})$ approximation for any $\eps>0$.
\end{lemma}
\begin{proof}
  By our assumption, there is a set $P$ of optimal hyperedges of size at least $dk/6$ such that every edge in $P$ intersects $K_1$. Consider two cases.

 Case 1: For at least half the edges $e\in P$, we have $|e\cap K_1|\geq 2$. Denote the set of these edges by $P'$. For every vertex $u\in V$, let its $K_1$-weight be the number of pairs $\{v,x\}$ such $v,x \in K_1$ and $\{u,v,x\}$ is a hyperedge. Then by our assumption, the vertices in $K$ have average $K_1$-weight at least $|P'|/k\geq d/12$. Choosing $2k/3$ vertices greedily (by maximum $K_1$-weight) gives (along with $K_1$) a $k$-subhypergraph with at least $dk/18$ hyperedges.

Case 2: $P''=P\setminus P'$ contains at least half the hyperedges in $P$. Note that $|e\cap K_1|=1$ for every $e\in P''$. For every pair of vertices $u,v\in V\setminus K_1$, let its $K_1$-weight be the number of vertices $x\in K_1$ such that $\{u,v,x\}$ is a hyperedge, and let $G$ be the graph on vertices $V\setminus K_1$ with these edge weights. Then any $k'$-subgraph of $G$ with total edge weight $w$ corresponds to a $(|K_1|+k')$-subhypergraph of $H$ with at least $w$ hyperedges, and in particular, $G$ contains a $k$-subgraph with average weighted degree at least $2|P''|/k\geq d/6$, which can be easily pruned (randomly or greedily) down to a $2k/3$-subgraph with average weighted degree $\Omega(d)$. Thus we can run the Densest $k$-Subgraph approximation algorithm of Bhaskara et al.~\cite{BCCFV10}\footnote{Strictly speaking, the algorithm in~\cite{BCCFV10} is defined for unweighted graphs, but one can easily adapt it by partitioning the edges into $O(\log n)$ sets with similar edge weights, and running the algorithm separately on every set of edges, thus losing only an additional $O(\log n)$ factor in the approximation.}, and find a $2k/3$-subgraph of $G$ with total weight at least $kd/n^{1/4+\eps}$, which in turn gives a $(|K_1|+2k/3=)k$-subhypergraph of $H$ with a corresponding number of hyperedges.
\end{proof}

In the more difficult case, at least half of the optimal hyperedges are fully contained in $H'$. Exploiting the
fact that the maximum degree in $H'$ is $\Delta$ and trading off multiple algorithms, we show in the following subsection
 how to obtain an $O(n^{\frac{4}{5}})$-approximation algorithm in this case.

\subsection{An $O(n^{4/5})$-approximation}

We start with a greedy algorithm similar to the greedy algorithm commonly used for Densest $k$-Subgraph~\cite{kp93, fkp, BCCFV10}.

\begin{algorithm}
 \KwData{3-uniform Hypergraph $H=(V, E)$, parameter $k$, vertex set $K_1 \subseteq V$ of size $k/3$}
For every $v\in V$, let its $K_1$-degree be $|\{e\in E\mid v\in e,e\cap K_1\neq\emptyset\}|$\;
 $K_2 \gets$ a set of $k/3$ vertices of highest $K_1$-degree ($K_1$ and $K_2$ may intersect)\;
 For any $u\in V$, let its $(K_1,K_2)$-degree be the number of edges of the form $(u,v,x)\in E$ such that $v\in K_2$ and $x\in K_1$\;
 $K_3 \gets$ a set of $k/3$ vertices of highest $(K_1,K_2)$-degree.  ($K_3$ may intersect $K_1$ and/or $K_2$)\;
 \KwRet{$K_1\cup K_2\cup K_3$}\;
 \caption{Greedy algorithm for Densest $k$-Subhypergraph in 3-uniform hypergraphs \label{alg:DkSH3-greedy1}}
\end{algorithm}
Algorithm~\ref{alg:DkSH3-greedy1} selects a subset $K_2$ of $k/3$ vertices $v$ with largest $K_1$-degree, i.e., 
the number of hyperedges incident to $v$ that contain at least one vertex of $K_1$. Then, a subset $K_3$ of $k/3$
vertices $w$ with largest $(K_1, K_2)$-degree is selected, where the $(K_1, K_2)$-degree of $w$ is the number of hyperedges containing $w$ 
of the form $\{w, x, y \}$ with $x \in K_1$ and $y \in K_2$. Note that the sets 
$K_1, K_2$ and $K_3$ are not necessarily disjoint and the returned
set may thus be smaller than $k$.

The following lemma gives a lower bound on the average degree guaranteed by this algorithm. It is a straightforward extension of similar algorithms for graphs. 

\begin{lemma}\label{lem:greedy1} Algorithm~\ref{alg:DkSH3-greedy1} returns a $k$-subhypergraph with average degree $\Omega(\Delta k^2/n^2)$.
\end{lemma}
\begin{proof}
By choice of $K_1$ and definition of $\Delta$, every vertex in $K_1$ has degree at least $\Delta$, and so the total number of edges containing vertices in $K_1$ is at least $\Delta |K_1|/3=\Delta k/9$ (since we could potentially be double-counting or triple-counting some edges).

If we were to choose $n$ vertices for $K_2$, there would be at least $\Delta k/9$ edges containing both a vertex in $K_1$ and a vertex in $K_2$ (as noted above). Choosing $k/3$ vertices greedily out of $n$ yields a set $K_2$ such that there are at least $\Delta k/9\cdot (k/3)/n=\Delta k^2/(27n)$ such edges.

Finally, choosing the $k/3$ vertices with the largest contribution (out of $n$) for $K_3$ ensures that there will be at least $\Delta k^2/(27n)\cdot (k/3)/n=\Omega(\Delta k^3/n^2)$ edges in $E\cap K_1\times K_2\times K_3$, giving average degree $\Omega(\Delta k^2/n^2)$.
\end{proof}

We now offer a second algorithm, which acts on $H'$ and is based on neighborhoods of vertices.

\begin{algorithm}
 \KwData{3-uniform Hypergraph $H'=(V',E')$ and parameter $k$.}
 \ForEach{vertex $v\in V$}{
 $G_v\gets(V\setminus\{v\},\{(u,x)\mid (v,u,x)\in E\})$\;
\ForEach{integer $\hat{d}\in[k-1]$}{
 $G^{\hat{d}}_v\gets G_v$\;
 \lWhile{there exists a vertex $u$ in $G^{\hat d}_v$ of degree $<\hat d$}{delete $u$ from $G^{\hat d}_v$}\;
 $S_v^{\hat d}\gets$ a set of $(k-1)/2$ vertices with highest degree in $G^{\hat d}_v$\;\label{step:choose-S}
 $T_v^{\hat d}\gets$ a set of $(k-1)/2$ vertices with the most neighbors in $S_v^{\hat d}$\;\label{step:choose-T}
}}
 \KwRet{The densest among all subhypergraphs $H'[\{v\}\cup S_v^{\hat d}\cup T_v^{\hat d}]$} over all choices of $v,\hat d$\;
 \caption{A neighborhood-based algorithm for Densest $k$-Subhypergraph in 3-uniform hypergraphs \label{alg:DkSH3-neighbor}}
\end{algorithm}

Algorithm~\ref{alg:DkSH3-neighbor} exploits the bound on the maximum degree in $H'$ to find a dense hypergraph inside the neighborhood of any vertex of degree $\Omega(d)$ in $K$, by considering the neighborhood of a vertex as a graph. Pruning low-degree vertices in this graph (which would not contribute many hyperedges to $K$) helps reduce the size of the graph, and makes it easier to find a slightly denser subgraph. Since the vertices of $K$ and their degrees are not known, the algorithm tries all possible vertices.

\begin{lemma}\label{lem:graph}
  If $H'$ contains a $k$-subhypergraph with average degree $d'=\Omega(d)$, then Algorithm~\ref{alg:DkSH3-neighbor} returns a $k$-subhypergraph with average degree $\Omega(d^2/(\Delta k))$.
\end{lemma}
\begin{proof}
  Since at the end of the algorithm we take the densest induced subhypergraph of $H'$ (among the various choices), it suffices to show that there is some choice of $v$ and $\hat d$ which gives this guarantee.  So let $v$ be an arbitrary vertex in $K$ with degree (in $K$) at least $d'$. We know that $G_v$ contains a subgraph with at most $k$ vertices and at least $d'$ edges, so its average degree is at least $2d'/k$. Setting $\hat d=d'/(2k)$, we know that the pruning procedure can remove at most $k\cdot d'/(2k)=d'/2$ out of the $d'$ edges in this subgraph, so the subgraph still retains at least $d'/2$ edges. On the other hand, we know that $G_v$ has at most $\Delta$ edges (since we've assumed the maximum degree in $H'$ is at most $\Delta$), and therefore, the same holds for the graph $G_v^{\hat d}$, in which the minimum degree is now at least $d'/2k$. This means that $G_v^{\hat d}$ has at most $2\Delta/(d'/2k)=O(\Delta k/d)$ vertices.

Since there exists a $k$-subgraph of $G_v^{\hat d}$ with $\Omega(d)$ edges, the greedy choice of $S_v^{\hat d}$ must give some set in which at least $\Omega(d)$ edges are incident. The greedy choice of $T_v^{\hat d}$ then reduces the lower bound on the number of edges by a $((k-1)/2)/|V(G_v^{\hat d})|=\Omega(d/\Delta)$ factor, giving us $\Omega(d^2/\Delta)$ edges. However, by the definition of $G_v$, together with $v$ these edges correspond to hyperedges in $H'$. Thus, the algorithm returns a $k$-subhypergraph with $\Omega(d^2/\Delta)$ hyperedges, or average degree $d^2/(\Delta k)$.
\end{proof}

Combining the various algorithms we've seen with a trivial algorithm and choosing the best one gives us the following guarantee:
\begin{theorem}\label{thm:DkSH-basic} There is an $O(n^{4/5})$-approximation for Dense $k$-Subhypergraph in 3-uniform hypergraphs.
\end{theorem}
\begin{proof}
By Lemma~\ref{lem:greedy2}, if at least half the optimal edges intersect $K_1$, then we can achieve a significantly better approximation (namely, $n^{1/4+\eps}$). 
Thus, from now on let us assume this is not the case. That is, $H'$ still contains a $k$-subhypergraph with average degree $\Omega(d)$. Again, recall that the maximum degree in $H'$  is at most $\Delta$.

By Lemma~\ref{lem:greedy1}, Algorithm~\ref{alg:DkSH3-greedy1} gives us a $k$-subhypergraph with average degree $d_1=\Omega(\Delta k^2/n^2)$. On the other hand, applying Algorithm~\ref{alg:DkSH3-neighbor} to $H'$ will give us a $k$-subhypergraph with average degree $d_2=\Omega(d^2/(\Delta k))$ by Lemma~\ref{lem:graph}.

Finally, we could choose $k/3$ arbitrary edges in $H$ and the subhypergraph induced on the vertices they span, giving us average degree $d_3\geq 1$. Thus, the best of the three will give us a $k$-subhypergraph with average degree at least $$\max\{d_1,d_2,d_3\}\geq(d_1^2d_2^2d_3)^{1/5}=\Omega((\Delta^2 k^4/n^4\cdot d^4/(\Delta^2k^2))^{1/5})=d\cdot\Omega((k^2/d)^{1/5}/n^{4/5}).$$
Since we must have $k^2/d\geq1$, the above gives an $O(n^{4/5})$ approximation.
\end{proof}

 \section{An improved approximation for 3-uniform Densest $k$-Subhypergraph} \label{sec:3DkSH-improved}

In Section~\ref{sec:DkSH3} we gave an $O(n^{4/5})$ approximation which combined a greedy algorithm 
 with Algorithm~\ref{alg:DkSH3-neighbor}, which looked for a dense subgraph inside a graph defined by the neighborhood of a vertex in $H$. To find this dense subgraph, we used a very simple greedy approach. However, we have at our disposal more sophisticated algorithms, such as that of Bhaskara et al.~\cite{BCCFV10}. One way to state the result in that paper (see Bhaskara's PhD thesis for details on this version~\cite{Bhaskara}) is as follows:
\begin{theorem}\label{thm:DkS} In any $n$-vertex graph $G$, for any $\alpha\in[0,1]$, if $k=n^\alpha$, then Densest $k$-Subgraph in $G$ can be approximated within an $n^{\eps}k^{1-\alpha}$ factor in time $n^{O(1/\eps)}$ for any $\eps>0$.
\end{theorem}
The $n^{1/4+\eps}$ guarantee of~\cite{BCCFV10} follows since for any $\alpha\in[0,1]$, we have $k^{1-\alpha}=n^{\alpha(1-\alpha)}\leq n^{1/4}$.

Using this guarantee instead of the simple greedy algorithm for D$k$S, we get the following improved algorithm for 3-uniform Densest $k$-Subhypergraph:
\begin{algorithm}
 \KwData{3-uniform Hypergraph $H'=(V',E')$ and parameters $k$ and $\eps>0$.}
 \ForEach{vertex $v\in V$}{
 $G_v\gets(V\setminus\{v\},\{(u,x)\mid (v,u,v)\in E\})$\;
\ForEach{integer $\hat{d}\in[k-1]$}{
 $G^{\hat{d}}_v\gets G_v$\;
 \lWhile{there exists a vertex $u$ in $G^{\hat d}_v$ of degree $<\hat d$}{delete $u$ from $G^{\hat d}_v$}\;
 $K^{\hat d}_v\gets$ the vertex set returned by the algorithm of Bhaskara et al.~\cite{BCCFV10} on the graph $G^{\hat d}_v$ with parameters $k-1$ and $\eps$\;
}}
 \KwRet{The densest among all subhypergraphs $H'[\{v\}\cup K_v^{\hat d}]$} over all choices of $v,\hat d$\;
 \caption{A D$k$S-based algorithm for Densest $k$-Subhypergraph in 3-uniform hypergraphs \label{alg:DkSH3-DkS}}
\end{algorithm}

The approximation guarantee in this final algorithm is given by the following lemma:
\begin{lemma}\label{lem:DkSH-DkS}
  Let $H'$ be an $n$-vertex 3-uniform hypergraph with maximum degree $\leq\Delta$, containing a $k$-subhypergraph of average degree $d'$, and let $\alpha,\beta$ be such that $k=n^\alpha$ and $\Delta k/d'=n^\beta$. Then Algorithm~\ref{alg:DkSH3-DkS} returns a $k$-subhypergraph of $H$ of average degree $$\Omega\left(\frac{d'}{n^{\eps+\alpha(2-\alpha/\min\{\beta,1\})}}\right).$$
\end{lemma}
\begin{proof}
As in the proof of Lemma~\ref{lem:graph}, we can deduce that for at least some choice of $v$ and $\hat d$, the graph $G_v^{\hat d}$ has at most $\min\{n,O(\Delta k/d')\}=O(n^{\min\{1,\beta\}})$ vertices and contains a $k$-subgraph with average degree $\Omega(d'/k)$.

By Theorem~\ref{thm:DkS}, since $k=n^{\alpha}=\Omega( |V(G_v^{\hat d})|^{\alpha/\min\{1,\beta\}})$, the algorithm of~\cite{BCCFV10} will return a $(k-1)$-subgraph of $G_v^{\hat d}$ with average degree $$\Omega\left(\frac{d'/k}{n^{\eps}k^{1-\alpha/\min\{\beta,1\}}}\right)=\Omega\left(\frac{d'}{n^{\eps+\alpha(2-\alpha/\min\{\beta,1\})}}\right).$$ As noted in the proof of Lemma~\ref{lem:graph}, this corresponds to a $k$-subhypergraph of $H'$ with the same guarantee.
\end{proof}

\begin{remark} In the notation of Lemma~\ref{lem:DkSH-DkS} we have $\Delta/d'=n^{\beta-\alpha}$ which implies that $\beta\geq\alpha$ (since $\Delta\geq d'$).
\end{remark}

Trading off the various algorithms we have seen, we can now prove the guarantee stated in Theorem~\ref{thm:3DkH-main}

\begin{theorem}[Theorem~\ref{thm:3DkH-main} restated] 
For every constant $\eps>0$, there is an $O(n^{4(4-\sqrt{3})/13+\eps}) \leq O(n^{0.697831+\eps})$-approximation for Densest $k$-Subhypergraph in $3$-uniform hypergraphs.
\end{theorem}
\begin{proof}

By Lemma~\ref{lem:greedy2}, if at least half the optimal edges intersect $K_1$, then we can achieve a significantly better approximation (namely, $n^{1/4+\eps}$). 
Thus, from now on let us assume this is not the case. That is, $H'$ still contains a $k$-subhypergraph with average degree $\Omega(d)$. Again, recall that the maximum degree in $H'$  is at most $\Delta$.

As before, let $\alpha,\beta$ be such that $k=n^{\alpha}$ and $\Delta k/d=n^{\beta}$. By Lemma~\ref{lem:greedy1}, Algorithm~\ref{alg:DkSH3-greedy1} gives us a $k$-subhypergraph with average degree $$d_1=\Omega(\Delta k^2/n^2)=\Omega\left(\frac{d}{(d/\Delta)n^2/k^2}\right)=\Omega\left(\frac{d}{n^{\alpha-\beta}n^{2-2\alpha}}\right)=\Omega\left(\frac{d}{n^{2-\alpha-\beta}}\right).$$ On the other hand, by Lemma~\ref{lem:DkSH-DkS}, Algorithm~\ref{alg:DkSH3-DkS} to $H'$ will give us a $k$-subhypergraph with average degree $$d_2=\Omega\left(\frac{d}{n^{\eps+\alpha(2-\alpha/\min\{\beta,1\})}}\right).$$

Let us analyze the guarantee given by the best of Algorithm~\ref{alg:DkSH3-greedy1} and Algorithm~\ref{alg:DkSH3-DkS}. First, consider the case of $\beta>1$. In this case, taking the best of the two gives us approximation ratio at most $n^{\eps+\min\{2-\alpha-\beta,\alpha(2-\alpha)\}}\leq n^{\eps+\min\{1-\alpha,\alpha(2-\alpha)\}}$. It is easy to check that this minimum is maximized when $\alpha=(3-\sqrt{5})/2$ giving approximation ratio $n^{(\sqrt{5}-1)/2+\eps}\leq n^{0.618034+\eps}$, which is even better than our claim.

Now suppose $\beta\leq 1$. In this case, the approximation guarantee is $n^{\eps+\min\{h_1,h_2\}}$, where $h_1=2-\alpha-\beta$ and $h_2=\alpha(2-\alpha/\beta)$. If $\alpha\geq 2/3$, then it can be checked that we always have $h_1\leq h_2$ for any $\beta\in[\alpha,1]$, in which case we have approximation factor at most $n^{\eps+2-2/3-2/3}=n^{2/3+\eps}$, which is again better than our claim. On the other hand, if $\alpha\leq (3-\sqrt{5})/2$, then $h_2\leq h_1$ for any $\beta\leq 1$, and so for this range of $\alpha$ we get approximation factor at most $n^{\eps+\alpha(2-\alpha)}\leq n^{(\sqrt5-1)/2}$, which as we've noted is also better than our claim. Finally, if $\alpha\in((3-\sqrt5)/2,2/3)$ then a straightforward calculation shows that $$\min\{h_1,h_2\}=\left\{\begin{array}{ll} h_1 &\text{if }\beta\geq 1-\frac{3\alpha}2+\sqrt{1-3\alpha+13\alpha^2/4}\\ h_2 &\text{otherwise,}\end{array}\right.$$ and that the value of $\min\{h_1,h_2\}$ is maximized at this threshold value of $\beta$. And so for $\alpha$ in this range we have $\min\{h_1,h_2\}\leq 1+\alpha/2-\sqrt{1-3\alpha+13\alpha^2/4}$, which is maximized at $\alpha=\frac{18+2\sqrt{3}}{39}\approx 0.55$, giving approximation ratio $n^{\eps+4(4-\sqrt{3})/13}$.
\end{proof}

\section{Minimum $p$-Union in 3-uniform hypergraphs} \label{sec:MpU3}

In this section we explore Minimum $p$-Union (the minimization version of Densest $k$-Subhypergraph), and give the following guarantee:
\begin{theorem}\label{thm:MkU} 
There is an $\tilde{O}(n^{2/5})$-approximation algorithm for Minimum $p$-Union in 3-uniform hypergraphs.
\end{theorem}

Note that this is significantly better than the $n^{0.69\ldots}$-approximation we would get by reducing the problem to Densest $k$-Subhypergraph via Theorem~\ref{thm:DkHMkU} and applying the approximation algorithm from Theorem~\ref{thm:3DkH-main}.

In this problem, we are given a 3-uniform hypergraph $H=(V,E)$, and a parameter $p$, the number of hyperedges that we want to find. Let us assume that the optimal solution, $P \subseteq E$, has $k$ vertices (i.e.~$|\cup_{e \in P} e| = k$). We do not know $k$, but the algorithm can try every possible value of $k=1,\ldots,n$, and output the best solution. Thus, we assume that $k$ is known, in which case the average degree in the optimum solution is $d=3p/k$. 

Recall that it is not necessary to get $p$ edges in one shot. By Lemma~\ref{lem:MkU-component}, 
 it is enough to find any subhypergraph of size at most $kn^{2/5}$ with average degree at least $\Omega(d/n^{2/5})$.

We follow along the lines of D$k$SH by choosing vertex set $K_1$ to be the $kn^{2/5}$ vertices of largest degree. The following lemma (corresponding to Lemma~\ref{lem:greedy2} for D$k$SH) shows that if at least half the edges in $P$ intersect $K_1$, then by Lemma~\ref{lem:MkU-component} we are done.

\begin{lemma}\label{lem:fix-MpU}
Suppose that at least half of the optimal edges contain a vertex of $K_1$. Then we can find a subhypergraph with at most $O(kn^{2/5})$ vertices and average degree at least $\Omega(d/n^{2/5})$.
\end{lemma}
\begin{proof}
  By our assumption, there is a set  of optimal hyperedges $P'\subset P$ of size at least $dk/6$ such that every edge in $P'$ intersects $K_1$.

As in the proof of Lemma~\ref{lem:greedy2}, if at least half the edges in $P'$ intersect $K_1$ in more than one vertex, then we can easily recover a set of $k$ vertices which along with $K_1$ contain at least $\Omega(p)=\Omega(kd)$ hyperedges. Since $|K_1|=kn^{2/5}$, this subgraph has $O(kn^{2/5})$ vertices and average degree $\Omega(d/n^{2/5})$ as required.

Thus, we may assume that at least half the edges in $P'$ intersect $K_1$ in exactly one vertex. Then again as in Lemma~\ref{lem:greedy2}, we define a graph $G$ on vertices $V\setminus K_1$ where every pair of vertices $u,v\in V\setminus K_1$ is an edge with weight $|\{x\in K_1\mid (u,v,x)\in E\}|$. Once again, subgraphs of $G$ with total edge weight $w$ correspond to a subhypergraphs of $H$ with at least $w$ edges, and in particular, $G$ contains a $k$-subgraph with average weighted degree at least $\Omega(d)$. Thus running the S$p$ES approximation of~\cite{CDK12} (or more precisely, the weighted version~\cite{DKN14}), gives a subgraph with at most $k f$ vertices and total edge weight at least $\Omega(kd)$ for some $f=n^{0.17+\eps}$ (which is well below $n^{2/5}$). Once again, the corresponding subhypergraph has at most $|K_1|+kf=O(kn^{2/5})$ vertices, and so the average degree is at least $\Omega(d/n^{2/5})$ as required.
\end{proof}

Thus, we will assume from now on that at least half of the hyperedges in $P$ \emph{do not} contain at least one vertex from $K_1$, i.e.~that $H'=H[V\setminus K_1]$ still contains at least half the hyperedges in $P$.  

As with D$k$SH, we now proceed with a greedy algorithm. Starting with the same vertex set $K_1$ defined above, it follows from Lemma~\ref{lem:greedy1} that if we run Algorithm~\ref{alg:DkSH3-greedy1} on $H$ with parameter $n^{2/5}k$, then we get a subhypergraph on $O(k n^{2/5})$ vertices induced on sets $K_1,K_2,K_3$ such that if the minimum degree in $K_1$ (which bounds the maximum degree in $V\setminus K_1$) is $\Delta$, then the subhypergraph has average degree $\Omega(\Delta k^2n^{4/5}/n^2)$. The total number of hyperedges in this subhypergraph is $\Omega(\Delta k^3n^{6/5}/n^2)=\Omega(\Delta k^3/n^{4/5})$. If this is at least $p=dk/3$, then we are done. Thus, we will assume from now on that $\Delta k^3/n^{4/5}=O(dk)$, that is 
\begin{equation}\label{eq:Delta-bound}
\Delta=O\left(\frac{dn^{4/5}}{k^2}\right).
\end{equation}

We reuse Algorithm~\ref{alg:DkSH3-neighbor} on $H'$, which gives us the following guarantee:

\begin{lemma}\label{lem:graph-min} Applying Algorithm~\ref{alg:DkSH3-neighbor} to the above hypergraph $H'$ with parameter $${\hat k}=\frac{k\sqrt{p\Delta}}{d}=\sqrt{\frac{k^3\Delta}{3d}}$$ returns a subhypergraph with at most $k f$ vertices and average degree at least $d/f$ for some $$f=O(\max\{k,n^{2/5}/\sqrt{k}\}).$$
\end{lemma}
\begin{proof}
As in the proof of Lemma~\ref{lem:graph}, we can deduce that for at least some choice of $v$ and $\hat d$, the graph $G_v^{\hat d}$ has at most $O(\Delta k/d)$ vertices and has minimum degree at least $\Omega(d/k)$. 

Note that we may not even have ${\hat k}$ vertices in $G^{\hat d}_v$. If we do have at least ${\hat k}$ vertices, then the greedy choice of $S^{\hat d}_v$ gives us $\Omega({\hat k}d/k)$ edges incident in the set (in fact, any choice of $\Omega({\hat k})$ vertices would do). The greedy choice of $T^{\hat d}_v$ then reduces the number of edges by (in the worst case) a ${\hat k}/(\Delta k/d)$-factor, giving us a total number of edges
$$\Omega\left(\frac{{\hat k}^{\,2}d^2}{\Delta k^2}\right)=\Omega(p).$$
Thus, in this case, we only need to bound the size of the subgraph. By~\eqref{eq:Delta-bound}, we can bound ${\hat k}$ as follows:
$${\hat k}=\sqrt{\frac{k^3\Delta}{3d}}=O\left(\sqrt{\frac{dn^{4/5}}{k^2}\cdot\frac{k^3}{d}}\right)=O\left(k\cdot\frac{n^{2/5}}{\sqrt{k}}\right),$$ which proves the lemma for this case.

If we do not have ${\hat k}$ vertices in $G^{\hat d}_v$, then the algorithm simply returns $G^{\hat d}_v$ itself, which has at most $\hat k=O(k\cdot n^{2/5}/\sqrt{k})$ vertices and average degree at least $\Omega(d/k)$, as required.

As noted in the proof of Lemma~\ref{lem:graph}, this corresponds to a subhypergraph of $H'$ with the same guarantee.
\end{proof}

We can now prove the main theorem.

\begin{proof}[Proof of Theorem~\ref{thm:MkU}]
By Lemma~\ref{lem:graph-min} and Lemma~\ref{lem:MkU-component}, it suffices to show that $\max\{k,n^{2/5}/\sqrt{k}\}=O(n^{2/5})$. 
 Since clearly $n^{2/5}/\sqrt{k}\leq n^{2/5}$, let us consider the parameter $k$. By definition of $d$ and $\Delta$, we clearly have $d\leq\Delta$, thus, by~\eqref{eq:Delta-bound} we have
 \begin{equation*}
d\leq\Delta=O\left(\frac{dn^{4/5}}{k^2}\right)
\end{equation*}
which implies $k=O(n^{2/5})$, and so the theorem follows.
\end{proof}

\section{Interval Hypergraphs} \label{sec:interval}
We show now that D$k$S and M$p$U can be solved in polynomial time on interval hypergraphs.
We only give an algorithm for M$p$U; a similar algorithm for D$k$S follows then from 
Observation~\ref{obs:poly-time-reduction}.

As defined in Section~\ref{sec:prelims}, a hypergraph $H = (V, E)$ is an interval hypergraph, if $V \subseteq \mathbb{N}$
and for each $e \in E$ there are integers $a_e, b_e$ such that $e = \{i \in V \, : \, a_e \le i \le b_e \}$.
Solving M$p$U on $H$ can be interpreted as finding $p$ intervals with minimum joint support.

\begin{theorem}
\emph{Minimum $p$-Union} is solvable in polynomial time on interval hypergraphs.
\end{theorem}
\begin{proof}
Let $b_1, ..., b_m$ be the largest elements in hyperedges $e_1, ..., e_m$ respectively, and assume that $b_i \leq b_j$ for any $i < j$.  Similarly let $a_1, ..., a_m$ be the smallest elements in $e_1, ..., e_m$ respectively.

We present a dynamic programming algorithm which calculates for each $j \leq i$ the optimal solution to an instance of Minimum $p$-Union on the hyperedges $e_1, ..., e_i$ with $p=j$ under the constraint that $e_i$ belongs to the solution.  Let $A[i, j]$ store the value of this optimal solution.  Assume that the values of $A$ have been computed for all  $i', j'$ with $j' \leq i' < i$.  We show how to compute $A[i, j]$ for any $j \leq i$.

We partition the hyperedges $e_1, ..., e_{i}$ in three sets $A_i, B_i, C_i$ with $A_i$ containing all hyperedges disjoint from $e_i$, $B_i$ containing all hyperedges intersecting but not included in $e_i$, and $C_i$ containing $e_i$ and all hyperedges included in $e_i$ (see Fig.\ \ref{fig:hyperedge_partition}).  Therefore we have: 
 \begin{enumerate}
	\item $b_{i'} < a_i$ for all $e_{i'} \in A_i$,
	\item $a_{i'} < a_i\leq b_{i'}$ for all $e_{i'} \in B_i$, and 
	\item $a_{i} \leq a_{i'} \leq b_{i'} \leq b_{i}$ for all $e_{i'} \in C_i$.
\end{enumerate}

Clearly, for every $j \leq |C_i|$ we have $A[i, j] = |e_i|$ since by definition of $A$, $e_i$ is included in the solution, and adding any other $j-1$ sets from $C_i$ to the solution does not increase the size of the union.  In the remainder of the proof, when we refer to an optimal solution corresponding to  $A[i', j']$ for some indices $i'$ and $j'$ we always mean a solution that uses the maximum number of sets in $C_{i'}$.

For any $t\geq 0$ and $j = t + |C_i|$, the optimal solution contains exactly $t$ sets in $A_i \cup B_i$.   Fix an optimal solution $OPT_i$ corresponding to $A[i, j]$ and let $e_{i^*}$ be the hyperedge with largest $b_{e_{i^*}}$ in $OPT_i$ that does not belong to $C_i$. We show that 
\begin{align}\label{eq:recursive}
A[i, j] = A[i^*, j - | C_i \setminus C_{i^*}| ] + |e_{i} \setminus e_{i^*}|.
\end{align}
Then, by considering every hyperedge with index $i' < i$ as the possible $i^*$ in Eq. (\ref{eq:recursive}) and taking the minimum value, one can compute $A[i, j]$ in linear time.

To complete the proof, we argue why Equation~\ref{eq:recursive} holds. First observe that a solution with 
value $A[i, j]$ exists. Indeed, by adding all elements of $C_i \setminus C_{i^*}$ to an optimal solution for  $A[i^*, j - | C_i \setminus C_{i^*}| ]$ we obtain a solution for $A[i, j]$ covering exactly $|e_{i} \setminus e_{i^*}|$ additional elements.  Next, assume that the value of $A[i, j]$ is less than that of Equation~\ref{eq:recursive}. Then we 
can obtain a solution for $A[i^*, j - | C_i \setminus C_{i^*}|]$ by removing from $OPT_i$ all the elements in $|C_i \setminus C_{i^*}|$ to obtain a solution with value at most $A[i, j] - |e_{i} \setminus e_{i^*}|$, contradicting the fact that $A[i^*, j - | C_i \setminus C_{i^*}| ]$ is the value of an optimal solution.
\end{proof}

 \begin{figure}
 \begin{center}
\definecolor{ffqqqq}{rgb}{1.0,0.0,0.0}
\definecolor{qqqqff}{rgb}{0.0,0.0,1.0}
\definecolor{qqffqq}{rgb}{0.0,1.0,0.0}
\begin{tikzpicture}[line cap=round,line join=round,>=triangle 45,x=1.0cm,y=1.0cm]
\clip(0.9,-0.6999999999999991) rectangle (12.940000000000005,5.240000000000003);
\draw [dotted] (1.0,3.0)-- (2.0,3.0);
\draw [dashed] (2.0,2.0)-- (8.0,2.0);
\draw [line width=2.0pt] (5.0,1.0)-- (12.0,1.0);
\draw [] (10.0,3.0)-- (11.0,3.0);
\draw [] (7.0,3.0)-- (9.0,3.0);
\draw [] (5.0,3.0)-- (6.0,3.0);
\draw [dotted] (3.0,3.0)-- (4.0,3.0);
\draw [rotate around={-1.4057492610463362:(4.530000000000001,3.16)},dashed] (4.530000000000001,3.16) ellipse (1.8960147050873561cm and 0.9676630415116041cm);
\draw (2.2,4.1) node[anchor=north west] {$C_{i'}$};
\draw (11.620000000000005,0.8200000000000014) node[anchor=north west] {$e_i$};
\draw (2.1800000000000006,2.0000000000000018) node[anchor=north west] {$e_{i'}$};
\begin{scriptsize}
\draw [fill=black] (1.0,3.0) circle (1.5pt);
\draw [fill=black] (2.0,3.0) circle (1.5pt);
\draw [fill=black] (2.0,2.0) circle (1.5pt);
\draw [fill=black] (8.0,2.0) circle (1.5pt);
\draw [fill=black] (5.0,1.0) circle (1.5pt);
\draw [fill=black] (12.0,1.0) circle (1.5pt);
\draw [fill=black] (10.0,3.0) circle (1.5pt);
\draw [fill=black] (11.0,3.0) circle (1.5pt);
\draw [fill=black] (7.0,3.0) circle (1.5pt);
\draw [fill=black] (9.0,3.0) circle (1.5pt);
\draw [fill=black] (5.0,3.0) circle (1.5pt);
\draw [fill=black] (6.0,3.0) circle (1.5pt);
\draw [fill=black] (3.0,3.0) circle (1.5pt);
\draw [fill=black] (4.0,3.0) circle (1.5pt);
\end{scriptsize}
\end{tikzpicture}
\vspace{-1.3cm}
\caption{Partitioning of hyperedges induced by $e_i$.  The dotted edges form set $A_i$, the dashed edge forms set $B_i$ and the elements of $C_i$ are represented by continuous edges.  The set $C_{i'}$ is also shown in dashed pattern.}
\label{fig:hyperedge_partition}
\end{center}
\end{figure}
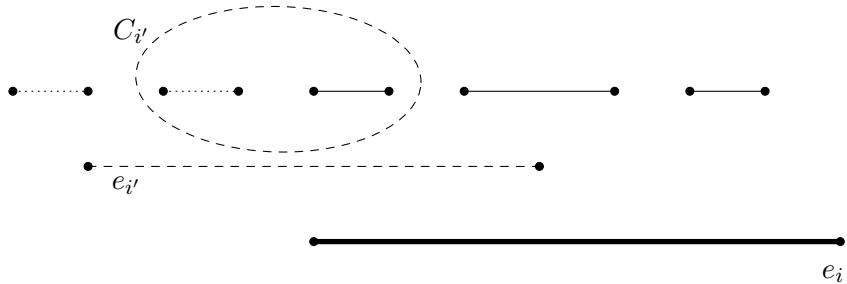

\section{Open problems} \label{sec:open}

While no tight hardness results are known for Densest $k$-Subgraph and Smallest $p$-Edge Subgraph, there are lower bounds given by the log-density framework~\cite{BCCFV10,CDK12}. In this framework, one considers the problem of distinguishing between a random graph and a graph which contains a planted dense subgraph. It has been conjectured that for certain parameters (namely, when the ``log-density" of the subgraph is smaller than that of the host graph), this task is impossible, thus giving lower bounds on the approximability of these problems. In the graph setting, the existing algorithm of~\cite{BCCFV10,CDK12} match these lower bounds.

However, in the hypergraph case, our current algorithms are still far from the corresponding lower bounds. In $c$-uniform hypergraphs, the lower bounds predicted by the log-density framework are $n^{(c-1)/4}$ for Densest $k$-Subhypergraph and $n^{1-2/(\sqrt{c}+1)}$ for Min $p$-Union. For $c=3$, for example, these lower bounds give $n^{1/2}$ and $n^{2-\sqrt{3}}=n^{0.2679\ldots}$, respectively (contrast with our current guarantees of $n^{0.6978\ldots}$ and $n^{0.4}$). The existing approach for the graph case does not seem to easily carry over to hypergraphs, and it remains a technical challenge to match the log-density based predictions for hypergraphs of bounded rank. 

For arbitrary rank, the lower bound given by the log-density framework is $m^{1/4}$ (note that we do not expect to achieve approximations that are sublinear in $n$ in this case), as opposed to our current guarantee of $\sqrt{m}$. In general hypergraphs, 
 one may also hope for hardness results which at the moment are elusive for the graph case or for bounded rank hypergraphs.

There is also an interesting connection between M$p$U/D$k$SH and the \emph{Small-Set Vertex Expansion} problem (SSVE)~\cite{AG11,LRV13,LM14}.  In Small-Set Vertex Expansion we are given a graph $G$ and a parameter $\delta$, and are asked to find the a set $V' \subseteq V$ with $|V'| \leq \delta n$ in order to minimize $\frac{|\{v \in V \setminus V' : v \in \Gamma(v)\}|}{|V'|}$.  Given a graph $G$, consider the collection of neighborhoods $\hat E = \{ \Gamma(v) : v \in V\}$ and the hypergraph $H = (V, \hat E)$.  If we let $p = \delta n$, the M$p$U problem (choosing $p$ hyperedges in $H$ to minimize their union) is quite similar to the SSVE problem.  The main 
difference is that 
 SSVE only ``counts" nodes that are in $V \setminus V'$, while M$p$U would also count nodes in $V'$.  It is known~\cite{Yury} that this special case of M$p$U reduces to SSVE, so it is no harder than SSVE, but it is not clear how much easier it is.  This motivates the study of M$p$U when hyperedges are neighborhoods in an underlying graph, and studying the approximability of this problem is an interesting future direction.

\bibliographystyle
{abbrv}
\bibliography{union}


\appendix

\section{Finding a Set of Minimum Expansion} \label{app:flow}
Given a bipartite graph $G=(E, V, F)$, the subroutine \textsc{Min-Exp}$(G)$ returns a subset of $E$ so that
$$\frac{|\textsc{Min-Exp}(G)|}{|\Gamma_G(\textsc{Min-Exp}(G))|} \ge \frac{|E'|}{|\Gamma_G(E')|}, $$
for every subset $E' \subseteq E$. 
Minimally expanding subsets of this kind have previously been used (e.g. in \cite{kr13,gkk12}) in communication
settings where computation time is disregarded. We therefore present a polynomial time implementation 
for \textsc{Min-Exp} using network flows. An alternative algorithm can be derived from a straightforward adaptation of a linear programming approach for the graph case due to Charikar~\cite{Charikar} to our setting (see Appendix~\ref{sec:charikar} for more details).

Let $N_q = (\tilde{G}, c_q, s, t)$ be a flow network with directed bipartite graph 
$\tilde{G} = (E \cup \{t \}, V \cup \{s \}, \tilde{F})$, capacities $c_q$ parameterized by a parameter $q$ 
with $\frac{m}{n} < q < m$, source $s$ and sink $t$ as follows (and as illustrated in Figure~\ref{fig:cut}):

\begin{enumerate}
 \item Vertex $s$ is connected to every $e \in E$ via directed edges (leaving $s$) with capacity $1$. 
 \item Every $v \in V$ is connected to $t$ via a directed edge (directed towards $t$) with capacity $q$.
 \item Edges from $F$ are included in $\tilde{F}$ and directed from $E$-vertex to $V$-vertex with capacity $\infty$.
\end{enumerate}

\begin{figure}
\begin{center}
  \includegraphics[height=6cm]{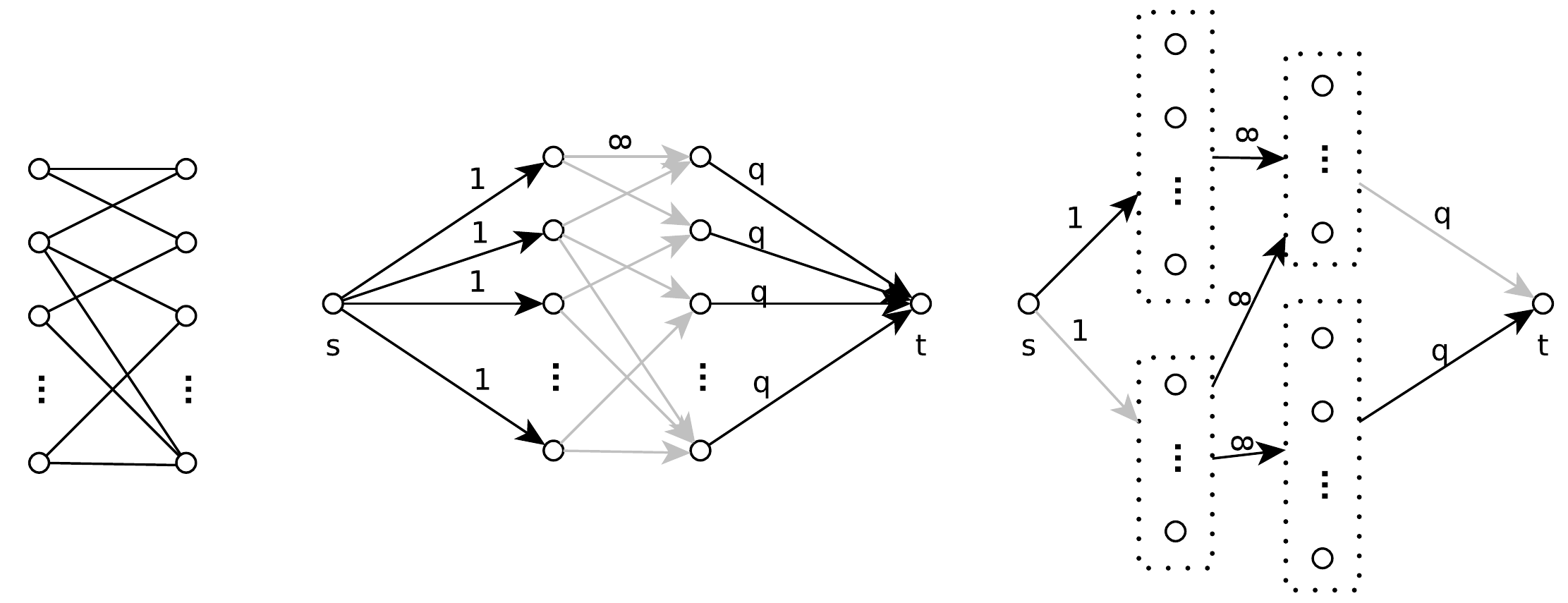}
  
  \vspace{-5.3cm}
  \hspace{9.35cm} $E_s \quad \quad \quad \quad \quad \quad \quad V_s$
  
  \vspace{3.8cm}
  
  Graph $G$ \hspace{3.3cm} Network $N_q$ \hspace{3.5cm} $E_t \quad \quad \quad \quad \quad \quad \quad V_t \quad \quad \quad$
    \vspace{0.2cm}
\end{center}
\caption{Left: Input graph $G$. Center: Network $N_q$. Right: Min-$s$-$t$-cut. Gray edges are cut edges.\label{fig:cut}}
\end{figure}

Denote by $C^*$ a minimum $s$-$t$ cut in $N_q$ and let $val(C^*)$ be the value of the cut. 
Since cutting all edges incident to vertex $s$ results in a cut of value $m$, the min-cut value is at most $m$ 
and thus finite, and, in particular, no edge connecting $E$ to $V$ is included in the min-cut. 
Denote by $E_s$ the set of $E$-vertices that, when removing the cut-edges from the graph, are incident to $s$,
and let $E_t = E \setminus E_s$. Let $V_s = \Gamma_G(E_s)$ and let $V_t = V \setminus V_s$. Since removing
$C^*$ from $\tilde{G}$ separates $s$ from $t$, all outgoing edges from $V_s$ are included in $C^*$. Furthermore,
since $C^*$ is a minimum cut, none of the edges leaving $V_t$ are contained in the cut. The resulting structure 
is illustrated on the right in Figure~\ref{fig:cut}. The value of the cut is computed as follows:
\begin{eqnarray}
 val(F^*) = |E_t| + q \cdot |V_s|. \label{eqn:392}
\end{eqnarray}

We prove now a property connecting the value of a minimum cut to the expansion of a subset of $E$.
This property allows us then to define an efficient algorithm for \textsc{Min-Exp}.

\begin{lemma}\label{lem:flow}
 Let $q$ be such that $\frac{m}{n} < q < m$. Then: 
 $$val(F^*) < m  \Leftrightarrow \exists E' \subseteq E: \frac{|E'|}{|\Gamma_G(E')|} > q.$$
\end{lemma}
\begin{proof}
Suppose that $val(F^*) < m$. We prove that $E' = E_s$ fulfills the claimed property. The value of the
cut $val(F^*)$ is computed according to Inequality~\ref{eqn:392} as follows:
\begin{eqnarray*}
 m > val(F^*) = |E_t| + q \cdot |V_s| = m - |E_s| +qp \cdot |V_s| = m - |E'| + q \cdot |\Gamma_G(E')|,
\end{eqnarray*}
which implies $\frac{|E'|}{|\Gamma_G(E')|} > q$ as desired.

 Suppose now that there is a $E' \subseteq E$ such that $\frac{|E'|}{|\Gamma_G(E')|} > q$. Then the set of edges
 $C$ consisting of those that connect $s$ to $E \setminus E'$ and those that connect $\Gamma_G(E')$ to $t$
 form a cut. We compute $val(C)$: 
 \begin{eqnarray*}
  val(C) & = & |E \setminus E'| + q |\Gamma_G(E')| = m - |E'| + q |\Gamma_G(E')| < m - |E'| + |E'| = m. 
 \end{eqnarray*}
The fact that $val(C^*) \le val(C)$ completes the proof. 
\end{proof}

Lemma~\ref{lem:flow} allows us to test whether there is a subset $E' \subseteq E$ such that 
$\frac{|E'|}{|\Gamma_G(E')|} > q$, for some value of $q$. For every set $E' \subseteq E$, we have
$\frac{|E'|}{|\Gamma_G(E')|} \in \{ \frac{a}{b} \, : \, a \in \{1, \dots, m \}, b \in \{1, \dots, n\} \}$. 
We could thus test all values $\frac{a}{b} - \epsilon$, for $a \in \{1, \dots, m \}, b \in \{1, \dots, n\}$ 
and a small enough $\epsilon$, in order to identify the desired set (or use a binary search to speed up the
process). Since computing a min-cut can be done in polynomial time, we obtain the following theorem:

\begin{theorem}
 Algorithm \textsc{Min-Exp} can be implemented in polynomial time.
\end{theorem}

\section{An LP-based algorithm for Minimum Expansion}\label{sec:charikar}

We use hypergraph notation in this section.  So the goal is to find a set $E' \subseteq E$ which minimizes $|\cup_{e \in Ee'} e| / |E'|$ over all choices of $E'$ (so there is no requirement that $|E'| = p$).

We use the following LP relaxation, which is a straightforward adaptation of Charikar's~\cite{Charikar} algorithm for graphs. 

\begin{align*}
  \mathrm{LP} = \min\quad&\sum_{i\in V}x_i\\
                      \text{s.t.}\quad&\sum_{e\in E}y_e=1\\
                      &x_i\geq y_e&\forall e\in E, i\in e\\
                      &x_i\geq 0 &\forall i\in V\\
                      &y_e\geq 0 &\forall e\in E\\
\end{align*}

Consider the following simple rounding algorithm:
\begin{itemize}
  \item Pick $r\in_R[0,1]$ uniformly at random.
  \item Let $E'=\{e\in E\mid x_e\geq r\}$.
  \item Let $V'=\bigcup_{e\in E_{\mathrm{ALG}}}e$.
\end{itemize}
Clearly, for every vertex $e\in E$ we have $$\prob[e\in E']=y_e.$$ Also, for every vertex $i\in V$ we have $$\prob[i\in V']=\max_{e\ni i}y_e\leq x_i.$$
Therefore, by linearity of expectation, we have $$\expec[\mathrm{LP}\cdot|E'|-|V'|]\geq \mathrm{LP}\cdot 1-\mathrm{LP}=0,$$ and this is obviously still true when we condition the expectation on $|E'|>0$ (a positive probability event), so with positive probability, we get a pair $(V_0,E_0)$ such that $E_0\neq\emptyset$, $V_0=\bigcup_{e\in E_0}e$ and $|V_0|/|E_0|\leq\mathrm{LP}$. The rounding is trivially derandomized by trying $r=y_e$ for every vertex $e\in E$.

\end{document}